\pdfoutput=1
\documentclass[sigconf]{acmart}

\usepackage{url,booktabs} 

\usepackage{schmpoly}

\newcommand{\cal}[1]{\mathcal{#1}}

\setcopyright{rightsretained} 

\begin{document}
\title{Schematic Polymorphism in the Abella Proof Assistant}

\author{Gopalan Nadathur}
\affiliation{%
  \institution{University of Minnesota}
}
\email{ngopalan@umn.edu}

\author{Yuting Wang}
\affiliation{%
  \institution{Yale University}
}
\email{yuting.wang@yale.edu}

\renewcommand{\shortauthors}{G. Nadathur and Y. Wang}

\begin{abstract}
The Abella interactive theorem prover has proven to be an effective
vehicle for reasoning about relational specifications. 
However, the system has a limitation that arises from the fact that it
is based on a simply typed logic: formalizations that are identical
except in the respect that they apply to different types have to be
repeated at each type. 
We develop an approach that overcomes this limitation while preserving
the logical underpinnings of the system.
In this approach object constructors, formulas and other relevant
logical notions are allowed to be parameterized by types, with the
interpretation that they stand for the (infinite) collection of
corresponding constructs that are obtained by instantiating the type
parameters.
The proof structures that we consider for formulas that are
schematized in this fashion are limited to ones whose type instances
are valid proofs in the simply typed logic. 
We develop schematic proof rules that ensure this property, a task
that is complicated by the fact that type information influences the
notion of unification that plays a key role in the logic. 
Our ideas, which have been implemented in an updated version of the
system, accommodate schematic polymorphism both in the core logic of
Abella and in the executable specification logic that it embeds. 
\end{abstract}


\settopmatter{printfolios=true}
\maketitle

\section{Introduction}\label{sec:intro}

The \Abella proof assistant~\cite{baelde14jfr} is a
vehicle for formalizing object systems that are described in a
syntax-directed and rule-based fashion.
Its success in this domain can be attributed to three key
characteristics.
First, it provides a means for constructing specifications
using relations that are given via {\it fixed-point
definitions}~\cite{mcdowell00tcs}, an approach that works well even in
the presence of non-deterministic and non-terminating behavior. 
Moreover, the fixed-point definitions can be specialized to yield a 
treatment of induction and co-induction~\cite{tiu04phd}.
Second, \Abella facilitates a {\it higher-order abstract syntax} treatment
of objects whose structures encompass bound variables.
It realizes this capability by providing $\lambda$-terms as a means
for representing objects, by permitting such terms to be compared
modulo $\lambda$-conversion, and by incorporating a special quantifier
$\nabla$ (pronounced as \emph{nabla}) that allows recursive
descriptions to encompass binding
constructs~\cite{baelde14jfr,miller05tocl}. 
Finally, \Abella supports a \emph{two-level logic} approach to
reasoning~\cite{gacek12jar,mcdowell02tocl}.
In this approach, object systems can be described in an expressive
and executable specification logic called the logic of hereditary
Harrop formulas or \HH~\cite{miller12proghol} and can subsequently be
reasoned about through an embedding in the main logic underlying
\Abella.
Since \HH is implemented in the language
\LProlog~\cite{miller12proghol}, this feature of \Abella enables a
transparent process of reasoning about actual 
programs, a capability that has, for example, been used in formalizing
compiler correctness arguments~\cite{wang16esopshort}.

The underlying basis for \Abella is provided by a logic called
\Gee~\cite{gacek11ic}.
This logic is built over the expressions of the simply typed
$\lambda$-calculus. 
As a consequence, the descriptions of objects, the definitions and the
theorems that can be constructed in \Abella are all monomorphically typed. 
This can sometimes make a formalization task more tedious and less
modular than is desirable.
For example, in some approaches to compiling functional programs, it
is necessary to consider multiple intermediate languages, the
expressions of which must be distinguished by their types.
However, correctness arguments need common operations, such as
combining lists of bound variables, and similar proofs of properties 
concerning such operations at each type.
In the current system, this leads to a replication of definitions and
proofs that differ only in the type to which they apply~\cite{wang16esopshort}.
This problem carries over also to the specification logic in the
two-level logic approach: the need to embed that logic within a simply
typed one means that it too must be simply typed and specifications or
implementations of ``library'' operations, such as those over lists,
have to be repeated at different types.

In this paper, we develop an approach to overcoming this difficulty
while retaining the logical basis of \Abella.
The core of our idea is simple: we parameterize
specifications, definitions, theorems and proofs by types such that we
can obtain the simply typed versions we want essentially by
instantiating the types.  
However, the actual realization of this idea is more subtle than might
initially be apparent. 
One issue arises from the need to embed schematized specifications
developed in \HH in the main logic: to support this possibility, it
becomes necessary to permit the schematization of {\it parts} of a
fixed-point definition in addition to allowing for the
parameterization of the entire definition by types. 
Another issue concerns the structure of proofs.
A schematic ``theorem'' can be part of a library only if its validity is
independent of the particular type vocabulary in existence at the time
that it is proved.
One way to enforce this requirement is to limit attention to only
those proofs for the theorem that are independent of the way in which
the types that parameterize the theorem are instantiated.
However, ensuring that this property holds is tricky because of a basic
characteristic of \Gee: type information significantly influences
unification over simply typed $\lambda$-terms, an operation that is
fundamental to supporting case-analysis style reasoning over
fixed-point definitions.

The rest of this paper is devoted to highlighting the issues discussed
above and to presenting the technical machinery that we have developed
towards solving them.
In the next section we introduce \Abella and the logic underlying it. 
In this presentation, we simplify the logic in a way that allows us to 
focus on the key aspects of our work without a loss of generality.
Section~\ref{sec:schm-lang} then describes a parameterization of the
core language with types; of particular interest here is the modified
form for fixed-point definitions and the interpretation of
parameterized theorems.
Section~\ref{sec:schm-proofs} presents a lifting of the proof
rules for the core logic in a way that ensures that the schematic
proofs that result from using the lifted proof rules  have the
properties we desire of them.  
Section~\ref{sec:schm-abella} illustrates the developments in the
previous two sections by showing how they can be used to construct
schematic proofs for some sample schematic theorems.
Section~\ref{sec:fulllogic} makes the argument that the technical
developments in this paper apply to the full \Abella system even
though we have exposed them in a simplified setting.
It does this by describing the aspects of the underlying logic that
were left out earlier and by outlining how the previously presented
ideas can be extended to cover them.\footnote{This discussion also
  makes the case that the addition of schematic polymorphism in fact
  rationalizes an aspect of the \Abella system that previously had an
  ad hoc treatment.}
We conclude the paper in Section~\ref{sec:comparison} with a
discussion of related work and an assessment of the enhancements
to \Abella that are made possible by this work.
We have in fact implemented our ideas in an enhanced version of
\Abella that supports schematic polymorphism.
Information on how to download this system can be found at the website
at \href{http://sparrow.cs.umn.edu/schmpoly}{\url{http://sparrow.cs.umn.edu/schmpoly}}.
This website also contains some examples that illustrate the use of
the new features of the system.

\section{The Abella Proof System}\label{sec:abella}

For most of this paper, we will limit our view of the logic underlying
\Abella to one that does not include co-induction and the $\nabla$
quantifier.  
We will also restrict the embedded specification logic to
that of Horn clauses, a sublogic of \HH.
Our presentation of \Abella in this section implicitly builds in
these simplifications.
We have chosen to limit our discussion in this way because it eases
the presentation of the core technical ideas in our work without
obscuring the issues that are essential to realizing schematic
polymorphism in the full system.
To validate the latter observation, we include in
Section~\ref{sec:fulllogic} a discussion of how the developments
described in the earlier parts of the paper are impacted by the
additional features of \Abella.

\subsection{A Logic with Fixed-Point Definitions}

The logic \Gee is based on an intuitionistic and predicative subset of
Church's Simple Theory of Types \cite{church40}.
One component of the type
vocabulary is a collection of {\it sorts} that includes \prop, the
type for propositions, and at least one other member.
In addition, the vocabulary may include type constructors, i.e. it may
contain symbols 
such as $c$ with a designated arity $n$ that can be used with
types $\alpha_1, \ldots, \alpha_n$ to yield the type $c\tyapp \alpha_1 
\tyapp \cdots \tyapp \alpha_n$.
The sorts and the types that are constructed in this way constitute
the \emph{atomic types}.
The \emph{function types} are constructed using the infix and right
associative operator $\rightarrow$.
Each well-formed type may obviously be written in the form $\alpha_1
\rightarrow \cdots \rightarrow \alpha_n \rightarrow \beta$ where
$\beta$ is an atomic type.
When a type is rendered in this form, we refer to
$\alpha_1,\ldots,\alpha_n$ as its argument types and 
to $\beta$ as its target type.
Note that the sequence of argument
types may sometimes be empty, i.e. the type may simply be of the form
$\beta$.
A predicate type is a type with at least one argument type
and a target type that is \prop.

Terms are constructed from collections of typed constants and variables
using abstraction and application, the latter being a left associative
operator.
The formation rules for terms are constrained by types and they also
associate types with well-formed expressions in the usual way. 
Two terms are considered equal and hence interchangeable if one can be
$\lambda$-converted to the other.
In light of this fact, and also the fact that such a form
exists for every term, we will assume that any term
that we examine is in $\lambda$-normal form.
The constants are sub-divided into non-logical and logical 
ones.
Non-logical constants are restricted in that \prop must not
appear in their argument types.
Such a constant is called a predicate if it has a predicate type.
The logical constants comprise $\top$ and $\bot$ of type \prop,
$\land$, $\lor$, and $\supset$ of type $\prop \ra \prop \ra \prop$,
and, for each type $\alpha$ that does not contain \prop,
$\forall_\alpha$ and $\exists_\alpha$ of type $(\alpha \ra \prop) \ra
\prop$.
Well-formed terms of the type \prop are also called
formulas.
Formulas (in normal form) have the structure
$(p\app t_1\app \cdots \app t_n)$ where $p$ is a constant. If $p$ is a
non-logical constant, the formula is said to be \emph{atomic} and it
has $p$ as its \emph{head}. 
The constants $\land$, $\lor$ and $\supset$, which correspond to the
familiar propositional connectives, are written in infix form 
and obey the usual precedence and associativity conventions.
The constants $\forall_\alpha$ and $\exists_\alpha$ constitute generalized
quantifiers: the formulas $\forall_\alpha \app (\lambdax{x}{B})$ and
$\exists_\alpha\app (\lambdax{x}{B})$ are also written as 
$\typedallx{\alpha}{x}{B}$ and $\typedsomex{\alpha}{x}{B}$ that
correspond to the more common form for quantified expressions.
The formula ${\cal Q}x_1:\alpha_1. \ldots {\cal
  Q}x_n:\alpha_n.F$ in which $\cal Q$ is $\forall$ or $\exists$ will
be abbreviated by ${\cal Q} {x_1:\alpha_1,\ldots,x_n:\alpha_n}.F$.
Another widely applied convention is to write a sequence of the form 
$x_1:\alpha_1,\ldots,x_n:\alpha_n$ in which $\alpha_1,\ldots,\alpha_n$
are identical as $x_1,\ldots,x_n:\alpha_1$. 

Derivability in \Gee is elaborated by means of a sequent calculus. The
actual formulation makes explicit the eigenvariables that are used in
the treatment of universal quantification.
Specifically, sequents take the form $\sequent{\Sigma}{\Gamma}{F}$
where $\Gamma$ is a multiset of formulas, $F$ is a formula and
$\Sigma$ is a collection of (typed) eigenvariables that might appear
in addition to the constants in the terms in $\Gamma$ and $F$.
Further, the quantifier rules are the following:

\smallskip

\begin{center}
\begin{tabular}{c}
\infer[\allL]
      {\sequent{\Sigma}{\Gamma,\typedallx{\alpha}{x}{B}}{F}}
      {\sequent{\Sigma}{\Gamma,\subst{t}{x}{B}}{F}}
\quad
\infer[\allR\; (x \notin \Sigma)]
      {\sequent{\Sigma}{\Gamma}{\typedallx{\alpha}{x}{B}}}
      {\sequent{\Sigma, x:\alpha}{\Gamma}{B}}

      \\[5pt]
\infer[\someL\; (x \notin \Sigma)]
      {\sequent{\Sigma}{\Gamma,\typedsomex{\alpha}{x}{B}}{F}}
      {\sequent{\Sigma, x:\alpha}{\Gamma,B}{F}}
\quad
\infer[\someR]
      {\sequent{\Sigma}{\Gamma}{\typedsomex{\alpha}{x}{B}}}
      {\sequent{\Sigma}{\Gamma}{\subst{t}{x}{B}}}
\end{tabular}
\end{center}

\smallskip

\noindent Note that the proviso on $x$ in the $\allR$ and 
$\someL$ rules can always be met via $\alpha$-conversion.
In the $\allL$ and $\someR$ rules, $t$ must be a term of type $\alpha$
that may use the symbols in $\Sigma$ in addition to the a priori
available constants and $B[t/x]$
represents the substitution of $t$ for $x$ in $B$; we assume here and
below that all substitutions are performed in a capture avoiding
way. 
The calculus also includes rules for the other (propositional) logical
constants and the usual structural rules for intuitionistic logic that
we assume the reader to be familiar with.

The logic \Gee actually represents a \emph{family} of logics in the sense
that it is parameterized by \emph{definitions}.
A definition is a  
collection of \emph{clauses} that each have the form
$\clause{x_1:\alpha_1\ldots x_n:\alpha_n}{A}{B}$, where $A$ is an
atomic formula all of whose free variables appear in $x_1,\ldots,x_n$
and have the respective types $\alpha_1,\ldots,\alpha_n$, and $B$ is a
formula all of whose free 
variables appear in $A$.
Such a clause is \emph{for} the predicate that is the head of $A$,
$x_1:\alpha_1,\ldots,x_n:\alpha_n$ is called its \emph{binder} and $A$
and $B$ respectively constitute its head and body.
Definitions must be constructed in \emph{blocks}, with all the clauses for a
particular predicate confined to one block.
Moreover, there is a \emph{stratification restriction} on definitions:
the body of a clause in a block may contain a predicate constant only
if it is defined in a previous block or in the current block
and, further, a predicate that is defined in the current block may not
appear in the antecedent of an implication in the body.
Variables in the binder of a clause can be renamed in the usual way. A
definition $\cal D'$ is a variant of another definition $\cal D$ if
the clauses in $\cal D'$ and $\cal D$ are identical but for such a
renaming. 
A definition is \emph{named away from} a set of (eigen)variables
$\Sigma$ if the variables in the binders of  its clauses are
distinct from those in $\Sigma$.

Let us assume that the definition parameterizing \Gee has been fixed
to be $\cal D$. In this context, an atomic formula is considered to
hold exactly when it is the head of an instance of a clause in $\cal
D$ whose body also holds.
This interpretation is realized through rules for introducing
atomic formulas on the left and right sides of a sequent.
In presenting these rules, we make use of the common view 
of a substitution as a mapping on variables that is then extended to
a mapping on terms and we write $\appsubst{\theta}{t}$ to denote
the application of a substitution $\theta$ to a term $t$. The rule for
introducing an atomic formula on the right, ``conclusion'' side of a
sequent is then the following:

\smallskip

\begin{center}
\begin{tabular}{c}
\infer[\defR]
      {\sequent{\Sigma}{\Gamma}{A}}
      {\sequent{\Sigma}{\Gamma}{\appsubst{\theta}{B}}}
\\[1pt]
$\clause{x_1 : \alpha_1 \ldots x_n : \alpha_n}{A'}{B} \in {\cal
  D}$ and $\appsubst{\theta}{A'} = A$
\end{tabular}
\end{center}

\smallskip

\noindent As should be readily apparent, this rule encodes the idea of
backchaining on a clause to complete a derivation.

The rule for introducing an atomic formula on the left side of a
sequent codifies a case analysis style of reasoning: to derive a
sequent that has an atomic formula as an assumption, we consider the
different ways in which the definition might cause an instance of the
formula to hold and show that the corresponding refinements of the
sequent are derivable. 
The following definition makes precise the cases that must be
considered towards this end.
Given two substitutions $\theta_1$ and $\theta_2$, we use the notation
$\substcomp{\theta_1}{\theta_2}$ here to denote the substitution that
produces the term $\appsubst{\theta_1}{(\appsubst{\theta_2}{t})}$ when
applied to the term $t$.
Further, we write $\appsubst{\theta}{\Gamma}$ to represent
the result of applying the substitution $\theta$ to each element of
sequence of formulas $\Gamma$.
Finally, if $\Sigma$ is a collection of variables, we write $\Sigma
\theta$ to denote the result of removing from $\Sigma$ the variables
in the domain of $\theta$ and adding the  
variables that appear in the range of $\theta$.

\begin{definition}\label{def:csu}
A \emph{complete set of unifiers} for two terms $t_1$ and $t_2$ of identical type
is a collection $\Theta$ of substitutions such that (a) for each
$\theta \in \Theta$ it is the case that
$\appsubst{\theta}{t_1}=\appsubst{\theta}{t_2}$, and (b) for any
substitution $\sigma$ such that $\appsubst{\sigma}{t_1} =
\appsubst{\sigma}{t_2}$, there is a $\theta \in \Theta$ and a
substitution $\rho$ such that $\sigma = \substcomp{\rho}{\theta}$.
We write \csu{t_1}{t_2} to (ambiguously) denote such a collection
$\Theta$. 
Then, given a sequent $\cal S$ of the form 
$\sequent{\Sigma}{\Gamma,A}{F}$ where $A$ is an atomic formula, and a
definition $\cal D$, $\defcases{\cal
  S}{\cal D}$ (ambiguously) denotes the following set of sequents: 
\begin{tabbing}
\quad\=\quad\=\kill
\>$\{ \sequent{\Sigma \theta}{\Gamma\theta, B\theta}{F\theta}\ \sep$ \\
\>\> $\clause{x_1 : \alpha_1\ldots x_n : \alpha_n}{A'}{B} \in \cal D\ \mbox{\it and}\ \theta \in \csu{A}{A'} \}$
\end{tabbing}
\end{definition}
\noindent The desired inference rule can be formalized using the above
definition as follows:

\smallskip

\begin{center}
\begin{tabular}{c}
\infer[\defL]
      {\sequent{\Sigma}{\Gamma, A}{F}}
      {\defcases{\sequent{\Sigma}{\Gamma, A}{F}}{\cal D'}}
\\[1pt]
${\cal D'}$ is a variant of ${\cal D}$ named away from $\Sigma$
\end{tabular}
\end{center}

\smallskip

\noindent Observe that this rule may lead to a multiplicity of cases being
considered for two different reasons: the head of a clause in $\cal D$
may unify with the atomic formula in the conclusion of the rule in
more than one way and more than one clause in the definition may unify
with this atom.  

When the clauses for a particular predicate in definition $\cal D$ are
finite in number, then this ensemble may also be designated as
\emph{inductive}.
The following  \emph{induction} rule may be used in
derivations if the clauses for $p$ are inductive:

\smallskip

\begin{center}
\begin{tabular}{c}
  \infer[\indL]
        {\sequent{\Sigma}{\Gamma, p\app \seq{t}}{F}}
        { \begin{array}{c}
            \{\sequent{\seq{x :\alpha}}{B[S/p]}{S\app \seq{t'}}\ \sep
            \clause{\seq{x : \alpha}}{p\app \seq{t'}}{B} \in {\cal D}\}
            \\
          \sequent{\Sigma}{\Gamma, S \app \seq{t}}{F}
          \end{array}
        }
\end{tabular}
\end{center}

\smallskip

\noindent The notation $\seq{\cdot}$ used here represents
sequences (or sets shown as sequences) of expressions whose
shape is indicated by $\cdot$.
The term that instantiates $S$ in a use of this rule is referred to as
the inductive invariant.
This term must contain no free variables and its type must be the same
as that of $p$.
The rule formalizes the intuition that if $p$
is defined inductively and $S$ satisfies all the clauses for $p$ then
$S\app \seq{t}$ must hold whenever $p\app \seq{t}$ does; hence 
$F$ must follow from $\Gamma$ and $p\app \seq{t}$ if 
it follows from $\Gamma$ and $S\app \seq{t}$. 

\subsection{Constructing Proofs in Abella}
\label{subsec:proofs-in-abella}

\Abella is a vehicle for interactively identifying
types, constants and definitions in \Gee and then trying to prove
relevant assertions.
Showing that a formula $F$ holds amounts to constructing a proof for
the sequent $\sequent{\cdot}{\cdot}{F}$.
\Abella supports a tactics-style approach based on the inference rules
of \Gee towards this end.
We illustrate the structure of this system through a few simple
reasoning examples.
In this discussion we focus mainly on the treatment of the definition
rules, assuming that the reader would be familiar with the treatment
of the remaining rules from other contexts. 

The examples that we consider concern representing and reasoning about
relations on lists.
We assume that the type vocabulary includes the two sorts $\iota$ and
\klist in addition to $\prop$.
We further assume that the availability of the nonlogical constants
\nil and $\scons$ of types \klist and $\iota \ra \klist \ra \klist$
respectively that provide us a means for constructing the commonly
used representations of lists in \Gee. 
To simplify notation, we will write $\scons$ as an infix and right
associative operator.
In this context let us consider a definition with the
following inductive clauses
for the predicate \gappend of type $\klist \ra \klist \ra \klist \ra
\prop$: 
\begin{tabbing}
\quad\=\qquad\qquad\=\kill
\>$\clause{l:\klist}{\gappend\app \nil\app l\app l}{\top}$\\
\>$\clauseheader{x:\iota,l_1:\klist,l_2:\klist,l_3:\klist}$\\
\>\>$\clausesans{\gappend\app (x\scons l_1)\app l_2\app (x\scons l_3)}
    {\gappend\app l_1\app l_2\app l_3}$
\end{tabbing}
As is perhaps evident, these clauses identify \gappend to be the
append relation on lists.
Once we have them, we can pose the usual logic programming style
queries about \gappend.
An example of such a query would be the question of whether the
formula $\rexistsx{l:\klist}{\gappend\app \nil\app  l\app l}$ holds.
A proof of the corresponding sequent can be constructed by using the
first clause for \gappend in a backchaining mode via the \defR rule to
derive the sequent obtained by instantiating the quantifier in the
formula with \nil. 

The fact that definitions have a fixed-point nature that goes beyond
the logic programming interpretation of clauses becomes clear when we
consider formulas in which atomic \gappend formulas appear negatively.
For example, consider proving
\begin{tabbing}
  \quad\=\kill
\>$\rforallx{l_1:\klist, l_2:\klist}{\gappend\app \nil\app l_1\app l_2
  \rimp l_1 = l_2}$;
\end{tabbing}
the equality relation used here represents $\lambda$-convertibility.
We can reduce this task through a few obvious steps to deriving 
the sequent
\begin{tabbing}
  \quad\=\kill
  \>$\sequent{(l_1:\klist,l_2:\klist)}{\gappend\app \nil\app l_1\app
    l_2}{l_1=l_2}$.
\end{tabbing}
Using the \defL rule will lead us to
consider the values for $l_1$ and $l_2$ under the different ways in
which $\gappend\app \nil\app l_1\app l_2$ can hold.
The head for only the first clause for \gappend unifies with this
formula and the proof can be concluded by noting that in this case
$l_1$ and $l_2$ must be identical. 

Many interesting properties need inductive arguments.
An example of such a property is the assertion that \gappend is
functional in its first two arguments:
\begin{tabbing}
\quad\=\qquad\qquad\=\kill
\>$\rforallxheader{l_1:\klist,l_2:\klist,l_3:\klist,l_4:\klist}$\\
\>\>${\gappend\app l_1\app l_2\app l_3 \supset \gappend\app l_1\app l_2\app l_4 \supset
  l_3 = l_4}.$
\end{tabbing}
While the \indL rule is the basis for inductive arguments,
using it directly usually leads to a complex construction.
\Abella provides a mechanism that is grounded in the \indL rule
but that yields more intuitive proofs.
This mechanism can be used to establish a formula of the form
$\rforallx{\overline{x:\alpha}}{F_1 \rimp \cdots \rimp F_n \supset
  C}$.
Specifically, it can be 
invoked with respect to some $F_i$ in this formula if $F_i$ is 
atomic and it has as its head a predicate whose clauses are marked as
inductive.
The result of doing so will be to add the formula to be proved as an
\emph{induction hypothesis} to the assumption set with the following
proviso: it can be used only by 
matching $F_i$ with an assumption formula that is obtained by
``unfolding'' the version of 
$F_i$ in the formula to be proved using a definition clause.
The proviso is realized by marking the predicate head of
$F_i$ in the unfolded form with $^*$ and in the original version with
$@$.
To illustrate the process, let us assume that $F$ represents the formula
\begin{tabbing}
\quad\=\qquad\qquad\=\kill
\>$\rforallxheader{l_1:\klist,l_2:\klist,l_3:\klist,l_4:\klist}$\\
\>\>${\gappend^*\app l_1\app l_2\app l_3 \supset \gappend\app l_1\app l_2\app l_4 \supset l_3 = l_4}$
\end{tabbing}
Then, using the \emph{induction tactic} and the rules for the logical
symbols appearing in the formula, we can reduce the task of showing
the functional property of \gappend to proving the following sequent:
\begin{tabbing}
\quad\=\qquad\qquad\=\kill
\>$\sequentheader{(l_1:\klist,l_2:\klist,l_3:\klist,l_4:\klist)}$\\
\>\>$\sequentsans{F, \gappend^@\app l_1\app l_2\app l_3, \gappend\app l_1\app l_2\app l_4}{l_3 = l_4}.$
\end{tabbing}
At this stage, we can use the \defL rule with respect to the second
assumption formula in the sequent.
The case when this formula is unfolded using the
first clause for \gappend has an easy proof.
Unfolding it using the second clause yields the sequent
\begin{tabbing}
\quad\=\qquad\qquad\=\kill
\>$\sequentheader{(x: \iota, l'_1,l_2,l'_3,l_4:\klist)}$\\
\>\>$\sequentsans{F, \gappend^*\app l'_1\app l_2\app l'_3, \gappend\app
            (x\scons l'_1)\app l_2\app l_4}{x \scons l'_3 = l_4};$
\end{tabbing}
note the changed annotation on \gappend that records the effect of
unfolding.
This sequent can be proved by using the \defL rule on the third
assumption formula and then ``applying'' $F$ to the remaining
assumption formulas. 

\subsection{Embedding a Separate Specification Logic}
\label{subsec:encode-spec-logic}

Relational specifications can be encoded and reasoned about directly
using definitions in \Gee. However, \Abella also allows specifications
to be written in an independent, executable specification logic and
then to be reasoned about through an encoding of that logic in
\Gee. We sketch the way in which this approach is realized below.

The specification logic is based, once again, on the simply typed
$\lambda$-calculus.
The expressions in this logic are required to be simply typed because
we desire eventually to embed them within \Gee.
The use of the (simply typed) $\lambda$-calculus, on the other hand, is
motivated by the benefits this choice provides in encoding formal
systems~\cite{miller12proghol}.
Types and terms in the logic are constructed as in \Gee with the
difference that $\omic$ is used for the type of propositions and the
logical constants are limited to $\strue$ of type $\omic$, $\sconj$ and
$\simply$ of type $\omic \ra \omic \ra \omic$ and, for each
(specification logic) type $\alpha$ not containing $\omic$,
$\sfall_\alpha$ of type $(\alpha \ra \omic) \ra \omic$.
The constants $\sconj$ and $\simply$ are the specification logic versions of
conjunction and implication and are written in infix form with the
usual precedence and associativity conventions.
The family of constants $\sfall_\alpha$ represent universal
quantification.
We will write $\sfall_\alpha\app (\lambdax{x}{F})$ as
$\typedforallx{x:\alpha}{F}$ and will abbreviate
$\typedforallx{x_1:\alpha_1}{\ldots\typedforallx{x_n}{\alpha_n}{F}}$
as $\typedforallx{x_1:\alpha_1,\ldots,x_n:\alpha_n}{F}$.
A term of type $\omic$ is called a 
(specification logic) formula.
A \emph{goal formula} is a formula that is atomic, $\strue$ or a
conjunction of goal formulas.
Relations are specified through a collection of \emph{definite
  clauses} that are closed formulas of the form
$\typedforallx{\seq{x:\alpha}}{G \simply A}$ where $G$ is a goal
formula and $A$ is an atomic formula.
The specification logic is oriented towards constructing derivations for a
closed goal formula $G$ from a specification $\Gamma$ which is a collection of
definite clauses, an objective
that is represented by the (specification logic) sequent
$\ssequent{\Gamma}{G}$.
The following inference rules may be used in realizing such an
objective: 
\begin{center}
\begin{tabular}{cc}

\qquad\quad \infer[\trueR]
      {\ssequent{\Gamma}{\strue}}
      {}

&

\infer[\andR]
      {\ssequent{\Gamma}{G_1 \sconj G_2}}
      {\ssequent{\Gamma}{G_1} \quad \ssequent{\Gamma}{G_2}}

\\[8pt]
\multicolumn{2}{c}{
\infer[\backchain \quad A\ \mbox{\rm is an atomic formula}]
      {\ssequent{\Gamma}{A}}
      {\ssequent{\Gamma}{\appsubst{\theta}{G}}}
}

\\[1pt]
\multicolumn{2}{c}{$\forallx{\seq{x:\alpha}}{G \simply A'} \in \Gamma,
  \theta\ \mbox{\rm is a closed substitution for}\ \seq{x}\mbox{\rm ,
  and}\ \appsubst{\theta}{A'} = A$}
\end{tabular}
\end{center}

\medskip
\noindent These rules can be attributed an execution semantics and, as
such, they are part of the mechanism for interpreting specifications as 
programs in \LProlog~\cite{miller12proghol}.

To encode the specification logic in \Gee, we need first to represent
its expressions.
This is accomplished by lifting its type and term vocabulary into
\Gee.
The specific collection of sorts, type constructors and non-logical
constants depends, of course, on the object system that we are
interested in formalizing.
The example that we consider below will again concern relations
on lists but will represent them this time in the specification
logic.
In this context, it will assume the sorts $\iota$ and \klist. 
We represent atomic formulas in a way that makes them syntactically
distinguishable: an atomic formula $A$ is represented by
$\atmconst\app A$ where $\atmconst$ is a constant of type $\omic \ra
\omic$. 

The derivation rules of the specification logic are encoded in \Gee as 
clauses that define a derivability predicate.
To support induction on the heights 
of derivations, this predicate is indexed by natural numbers, 
which are represented by terms of type \typenat constructed using the
constants \zero and \s of type \typenat and $\typenat \ra \typenat$
respectively.
In the context of \Gee, the basis for induction is 
provided by the following inductive clauses for the predicate \termnat
that has the type $\typenat \ra \prop$:
\begin{center}
  \begin{tabular}{cc}
  $\clausesans{\termnat\app \zero}{\top}$ \qquad & \qquad
$\clause{n:\typenat}{\termnat \app (\s\app n)}{\termnat\app n}$ 
\end{tabular}
\end{center}
The derivability predicate \prove has the type $\typenat \ra
\omic \ra \prop$ and is defined by the following clauses:
\begin{tabbing}
\quad\=\quad\=\kill
\>$\clause{n:\typenat}{\prove\app n\app \strue}{\top}$\\
\>$\clauseheader{n : \typenat, g_1 : \omic, g_2 : \omic}$\\
\>\>$\clausesans{\prove\app (\s\app n)\app (g_1 \sconj g_2)}
    {(\prove\app n\app g_1) \rand (\prove\app n \app g_2)}$\\
\>$\clauseheader{n:\typenat,a:\atm}$\\
\>\>$\clausesans{\prove\app (s\app n)\app (\atmconst\app a)}
    {\typedsomex{\omic}{g}{(\prog\app a \app g) \rand (\prove \app n\app g)}}$
\end{tabbing}
Here, the predicate \prog, which has type $\omic \ra \omic \ra \prop$,
serves as the interface for encoding particular specifications that
might be provided in the specification logic.
For example, suppose that we have identified $\append$ of type $\klist
\ra\klist\ra \klist\ra \omic$ to be the specification logic encoding
of the append relation on lists through the following clauses:
\begin{tabbing}
\quad\=\quad\=\kill
\>$\forallx{l:\klist}{\strue \simply (\append\app \nil\app l\app l)}$\\
\>$\forallxheader{x:\iota}\forallxheader{l_1,l_2,l_3:\klist}$\\
\>\>${\append\app l_1\app l_2\app l_3
      \simply \append\app (x\scons l_1)\app l_2\app (x \scons l_3)}$;
\end{tabbing}
\nil and $\scons$ are used here as the specification logic versions
of the list constructors that are then lifted into \Gee.
This definition translates into the following clauses for \prog in \Gee:
\begin{tabbing}
\;\;\=\;\;\=\kill
\>$\clause{l:\klist}{\prog\app (\append\app \nil\app l\app l)\app \strue}{\top}$\\
\>$\clauseheader{x:\iota,l_1:\klist, l_2:\klist, l_3:\klist}$\\
\>\>$\clausesans{\prog\app (\append\app (x\scons  l_1)\app l_2\app (x \scons l_3)) \app
   (\atmconst\app (\append\app l_1\app l_2\app l_3))}{\top}$
\end{tabbing}

Specifications can be constructed in a standalone mode in the
specification logic and can be used as programs to realize
computations.
When we want to reason about such a program, we begin the process by
\emph{loading} the corresponding specification into \Abella.
Doing so lifts the vocabulary the specification introduces into \Gee
and adds a \prog clause for each of the definite clauses it contains.
At this stage, we can prove properties of the program by using the
\prove predicate to represent what is derivable from the program in
the specification logic.
A special notation is provided to make the encoding transparent: the
expression $\sprove F$ represents the \Gee formula
$\rexistsx{n:\typenat}{nat\app n \rand \prove\app n\app F}$.
Thus, the functional nature of the rendition of \append in the
specification logic can be expressed via the \Gee formula
\begin{tabbing}
  \quad\=\quad\=\kill
  \>$\rforallxheader{l_1,l_2,l_3,l_4:\klist}$\\
  \>\>${{\sprove{\append\app l_1\app l_2\app l_3} \supset
    \sprove{\append\app l_1\app l_2\app l_4} \supset l_3 = l_4}}.$
\end{tabbing}
Applying the \defL rule to the $\sprove{\cdot}$ predicate translates
transparently into case analysis over the specification logic
clauses.
For example, using it on the formula $\sprove{\append\app l_1\app
  l_2\app l_3}$ will yield the two cases corresponding to
$\sprove{\append\app nil\app l\app l}$ (setting $l_2$ and $l_3$ 
to $l$) and $\sprove{\append\app (x \scons l'_1)\app l_2\app (x\scons
  l'_3)}$ (setting $l_1$ to $x \scons l'_1$ and $l_3$ to $x \scons
l'_3$).
The induction tactic can also be invoked with respect to the
$\sprove{\cdot}$ predicate.
Such an application amounts eventually to an induction on
the height of a derivation measured by the first argument to $\prove$.
It is an easy exercise to construct a proof of the functional property
of \append using these observations. 

\section{A Parameterization by Types}
\label{sec:schm-lang}

In the examples pertaining to lists that we considered in the previous
section, we fixed the type of list elements to be $\iota$.
Having to fix the type in this fashion is a consequence of using a
simply typed language.
In an actual application, we may need to consider relations both in
the specification logic and in \Gee over lists with different types of
elements.
The formalization of these relations would have a similar structure as
also would the proofs of properties pertaining to them.
However, the \Abella system in its current form requires us to repeat
the formalization and the proofs for each type at which they are needed.
Our goal in this paper is to provide a means for avoiding this kind of
redundancy.
We aim to do this by enabling the process to be parameterized by
types in a way that ensures that instantiating the parameters with actual
types will yield constructions that are legitimate in the underlying
(simply typed) logic. 
We begin the task of realizing our goal by describing in this section
a modification to the language and the structure of definitions that
provide the basis for the desired parameterization.

The first step in this direction is to allow variables to appear as
atomic types in type expressions.
We call type expressions not containing type variables \emph{concrete
  types} or \emph{ground types}.
%
A \emph{type schema} is an expression of the form
$(\tyschm{A_1,\ldots,A_n}{\tau})$ where $A_1,\ldots,A_n$ is a sequence
of distinct type variables and $\tau$ is a type expression all of
whose type variables appears in $A_1,\ldots,A_n$.
A term constant now has a type schema associated with it.
We indicate such an association by writing
$(c:\tyschm{A_1,\ldots,A_n}{\tau})$. 
For example, a parameterized development pertaining to lists would use
the constants $(\nil:\tyschm{A}{\klist\ A})$ and $(\scons:\tyschm{A}{A \ra
  \klist\tyapp A \ra \klist\tyapp A})$, where $\klist$ is now a unary
type constructor.\footnote{This example also shows the usefulness of
  type constructors in the parameterization.}
An \emph{instance} of a constant of this kind is obtained by substituting
the types $\tau_1,\ldots,\tau_n$ for $A_1,\ldots,A_n$ and its type is
the corresponding instance of $\tau$.
We denote such an instance by $\cinst{c}{\tau_1,\ldots,\tau_n}$.
Thus, $\cinst{\nil}{\iota}$ represents an instance of $\nil$ that has
the type $(\klist\tyapp \iota)$.
In what follows, we may drop the subscript in depicting the instance
of a constant if the information it provides is not important
to the discussion.  
%
%

Terms are formed in the same way as in \Gee, except that we now use
instances of constants and type variables may appear in the types
of term variables.
A noteworthy aspect in this context is that the schematization of
types allows us to view the quantifier symbols differently from
before.
Instead of their being distinct constants at each relevant type, we
view them as the ``schematic'' constants $(\forall : \tyschm{A}{(A \ra
  prop) \ra \prop})$ and $(\exists : \tyschm{A}{(A \ra prop) \ra
  \prop})$  from which desired instances can be generated.
Observe that the types appearing in a term may now contain type variables.
A well-formed type- or term-level expression is said to be well-formed
  \emph{relative to a set of type variables $\Psi$} if the type
  variables appearing in the expression are contained in $\Psi$.
We extend the notion of substitutions to also include mappings on types.
We write $\apptysubst{\phi}{e}$ to denote the application of a
type-level substitution $\phi$ to a type-level or term-level
expression $e$ and $\tysubstcomp{\phi_1}{\phi_2}$ to denote the
substitution that produces
$\apptysubst{\phi_1}{\apptysubst{\phi_2}{e}}$ when applied to any term
or type $e$. 
A term in the extended language represents the
collection of terms in \Gee that are obtained by substituting ground
types for its type variables. 

We now consider parameterizing definitions by types.
Such a parameterization must occur at least at the \emph{level of
  definition blocks}. 
For example, consider the definition of \gappend from the previous
section.
In the schematic version, this definition should be for an \gappend
with the type $(\klist\tyapp A) \ra (\klist\tyapp A) \ra
(\klist\tyapp A) \ra \prop$. 
Note also that we would want the schematization to work in such a way
that each type instance that we generate of the definition block for
\gappend to be independent of all other instances of the same block.
A less obvious but still important observation is that we will need
type parameterization to work also at the \emph{clause level}, \ie, we
will want a concise, schematic way to signal the inclusion of all the
type instances of a clause within a \emph{single} definition block.
This capability is needed to accommodate schematic polymorphism in the
specification logic. 
To understand this, consider the definite clauses in the specification logic
that we saw for the predicate \append.
To make these work for lists of different types, we would have to
associate a type of the form $(\klist\tyapp A) \ra
  (\klist\tyapp A) \ra (\klist\tyapp A) \ra \omic$ with \append and we
would need to treat the definite clauses as if they represent the
collection of all their ground type instances.
This kind of an interpretation is already built at a computational level
into the language \LProlog.
To be able to \emph{reason} about such a specification, we would have
to embed the collection of all the instances of these definite clauses within
\Gee.
Under the embedding scheme that we have described in the previous
section, this means that we would have to be able to parameterize the
\prog clauses that encode each of the definite clauses for \append by
the type of the list elements.  

The above discussion motivates the following definitions.
\begin{definition}\label{def:schm_def_block}
A \emph{schematic clause} is an expression of the form
$\schmclause{\Psi}{\seq{x:\alpha}}{A}{B}$, where
$\clause{\seq{x:\alpha}}{A}{B}$ has the structure of a clause in \Gee
except that variables may appear in the types and $\Psi$ is a listing
of some of the type variables appearing in the clause.
All the terminology associated with clauses in \Gee carries over to
the schematic version.
A \emph{schematic definition block} comprises a finite sequence of
distinct type variables $\Psi'$, a finite set of predicate constants
$\{c^1:\tyschm{\Psi'}{\tau_1},\ldots,c^n:\tyschm{\Psi'}{\tau_n}\}$,
and a collection of schematic clauses each of which is
for some ${c^i}$.
A schematic clause and definition block of the forms shown are said to
be \emph{parameterized} by $\Psi$ and $\Psi'$ respectively. 
\end{definition}
%
%
%
\begin{definition}\label{def:schm_block}
A schematic definition block parameterized by $\Psi'$ with
associated predicate constants
$\{c^1:\tyschm{\Psi'}{\tau_1},\ldots,c^n:\tyschm{\Psi'}{\tau_n}\}$ is
\emph{well-formed} if, for every clause
$\schmclause{\Psi}{\seq{x:\alpha}}{A}{B}$ in it, it is the case that (a)~$\Psi$ is
disjoint from $\Psi'$, (b)~$A$ and $B$ are well-formed in $\Psi \cup
\Psi'$, (c)~all occurrences of $c^i$ in $A$ and $B$ are at the instance
$\cinst{c^i}{\Psi'}$, and (d)~all the type variables that occur in $B$ also
occur in $A$.
\end{definition}
\noindent The wellformedness definition above requires 
the constants $c^1, \ldots,$ $c^n$ introduced  by a schematic
definition block to be used at their ``defined types''
at every occurrence in the block.
This condition ensures the independence of each instance of such a
block from every other instance. 
If $\cal B$ is a schematic definition block, we will write
$\binst{\cal B}{\seq{\tau}}$ to represent an instance of the block
that is obtained by substituting the type expressions $\seq{\tau}$ for
the type variables that parameterize $\cal B$.
Subsequent to the block,
each $c^i$ is treated as having the type schema $[\Psi']\tau_i$
associated with it.
%
The requirement that all the type variables in the body of a
schematic clause also occur in the head has a technical motivation: it 
ensures that the type instance of the body is fixed as soon as the
type instance of the head is determined, a property that will become
important when we consider the construction of proofs.

A schematic definition block serves as an abbreviated
representation of a collection of definition blocks in \Gee that are
obtained as follows.
First, we instantiate the 
type variables that parameterize the block with concrete types.
Within the structure thus obtained, we generate all the versions of
each schematic clause by instantiating the type variables that
parameterize the clause with all available concrete types.
Note that both the collection of definition blocks and the collection of
clauses within each block that are generated in this way are sensitive
to the vocabulary of types in existence at a particular point.
However, the schematic proofs whose construction we will support will
be such that they will allow us to prove only those statements whose
instances have derivations in \Gee independently of the
  available type signature.

We adopt also a schematic view of the properties we would like to
prove in the context of a schematic definition.
A \emph{schematic formula} is an expression of the form
$\schmthm{A_1,\ldots,A_n}{F}$ in which $F$ is a well-formed formula
relative to the collection of type variables $A_1,\ldots,A_n$.
We say that such a formula is parameterized by the type variables
$A_1,\ldots,A_n$.
Given the ground types $\tau_1,\ldots,\tau_n$, we can generate the
formula $\apptysubst{\tau_1/A_1,\ldots,\tau_n/A_n}{F}$ in \Gee.
A schematic formula is considered a \emph{schematic theorem} only when
any of its type instances generated in this way is a theorem of \Gee. 
A schematic theorem not parameterized by any type variable
coincides with a theorem in \Gee.

\section{Schematizing Proofs}
\label{sec:schm-proofs}

Schematic theorems must be such that their type instances hold in
\Gee regardless of the type vocabulary.
While there can be different approaches to establishing such theorems,
our focus here will be on what we call \emph{schematic proofs}.
These are structures associated with schematic formulas 
that yield proofs in \Gee of type instances of the formulas simply
by instantiating the type variables that appear in them.
For us to be able to generate such a structure, it must be the case
that a type instance of the schematic formula has a proof in \Gee that
does not use information specific to that instance.
This is true of many proofs in \Gee.
For example, the proofs we sketched in Section~\ref{sec:abella} for
the different properties of \gappend did not depend on the element
type being $\iota$.
Our objective in this section is to lift the proof rules for \Gee to
apply to schematic formulas in such a way that they yield schematic
proofs.
The main challenge in realizing this objective is articulating a
schematic version of the \defL rule: unification plays a 
fundamental role in the formulation of this rule and unification over
simply typed $\lambda$-terms depends significantly on type
information~\cspc\cite{nadathur92typesshort}. 

\subsection{Schematizing sequents and the basic rules}
\label{subsec:schm-basic-rules}

The lifted versions of our proof rules will apply to \emph{schematic
  sequents} that have the form $\schmsequent{\Psi}{\Si}{\G}{B}$.
These sequents augment the ones in \Gee 
with a set $\Psi$ of type variables that binds the type variables in
$\Si$, $\G$ and $B$.
This set will remain unchanged throughout the
derivation of the sequent.
Thus, the variables in this set will
function as placeholders for arbitrary types but will be like ``black
boxes'' in that they will not allow us to look at or use the
particular structures of the types that fill them.
To prove a schematic formula $[A_1,\ldots,A_n]F$, we will need to derive
the schematic sequent
$\schmsequent{A_1,\ldots,A_n}{\emptyset}{\emptyset}{F}$. 

The schematic versions of the rules for the logical symbols in \Gee
are obtained essentially by adding a set of type variables,
represented by a schema variable such as $\Psi$, to the sequents that
form the premises and conclusions of the rules in \Gee.
For example, the following constitute schematic versions of the
quantifier rules:

\smallskip

\begin{center}
\begin{tabular}{c}
\infer[\schmallL]
      {\schmsequent{\Psi}{\Sigma}{\Gamma,\typedallx{\alpha}{x}{B}}{F}}
      {\schmsequent{\Psi}{\Sigma}{\Gamma,\subst{t}{x}{B}}{F}}
\quad
\infer[\schmsomeR]
      {\schmsequent{\Psi}{\Sigma}{\Gamma}{\typedsomex{\alpha}{x}{B}}}
      {\schmsequent{\Psi}{\Sigma}{\Gamma}{\subst{t}{x}{B}}}
\\[5pt]
\qquad \infer[\schmallR\; (x \notin \Sigma)]
      {\schmsequent{\Psi}{\Sigma}{\Gamma}{\typedallx{\alpha}{x}{B}}}
      {\schmsequent{\Psi}{\Sigma, x:\alpha}{\Gamma}{B}}
\\[5pt]
\qquad \infer[\schmsomeL\; (x \notin \Sigma)]
      {\schmsequent{\Psi}{\Sigma}{\Gamma,\typedsomex{\alpha}{x}{B}}{F}}
      {\schmsequent{\Psi}{\Sigma, x:\alpha}{\Gamma,B}{F}}
\end{tabular}
\end{center}

\smallskip

\noindent As before, $t$ must be a term of type $\alpha$ in these
rules and it must be constructed using only the available constants
and the symbols in $\Sigma$. 
It is easy to see the schematic nature of these rules: by
instantiating the premises and conclusions of each rule with any
substitution of ground types for variables in $\Psi$, we get a rule in
\Gee.
The rules for the remaining logical symbols have a similarly obvious
structure and quality.

\subsection{Schematic rules for fixed-point definitions}
\label{subsec:schm-def-rules}

The schematic forms of the rules for introducing atomic formulas in
sequents have to pay more careful attention to the interpretation of
type variables.
The following definition will be useful in formalizing these rules. 
\begin{definition}\label{def:red_schm_clause}
The \emph{reduced form of a schematic clause} $[\Psi]C$ that appears in a schematic
definition block parameterized by $\Psi'$ is $[\Psi'']C$ where
$\Psi''$ contains type variables from $\Psi$ and $\Psi'$ only if they
occur in $C$. Note that because the type variables in the body of the
clause must occur in its head, $\Psi''$ contains exactly the type
variables that occur in the head of $C$. We write $\reduceddef{\cal
  D}$ to represent a collection of the reduced forms of the schematic
clauses in the definition $\cal D$. 
\end{definition}

The main issue in formalizing the schematic version of the \defR rule,
which we denote by \schmdefR, is that we have to consider
instantiating the type variables in schematic clauses. This rule is
presented below.

\smallskip

\begin{center}
\begin{tabular}{c}
\infer[\schmdefR]
      {\schmsequent{\Psi}{\Sigma}{\Gamma}{A}}
      {\schmsequent{\Psi}{\Sigma}{\Gamma}{\apptysubst{\phi}{\appsubst{\theta}{B}}}}
\\[1pt]
$\schmclause{\Psi'}{\seq{x : \alpha}}{A'}{B} \in {\reduceddef{\cal D}}\ 
         \mbox{\rm and}\ \apptysubst{\phi}{\appsubst{\theta}{A'}} = A$
\end{tabular}
\end{center}

\smallskip

\noindent In this rule $\phi$ and $\theta$ are, respectively, a type
substitution for $\Psi'$ and a term substitution for $\seq{x}$ whose
range is well-formed with respect to $\Psi$. 

As we have already noted, the schematization of the \defL rule is
complicated by the fact that the structure of unifiers for terms in
the simply typed $\lambda$-calculus can depend on types. 
To circumvent this difficulty, we will limit ourselves to those situations
in which the \emph{complete set of unifiers (CSUs) have the
  same structure no matter how type variables are instantiated}.
At a practical level, such unifiers can be computed, for instance, by
using higher-order pattern unification that does not pay attention to
type information~\cspc\cite{miller91jlc,nipkow93lics}. 
Note that the described constraint on CSUs must hold independently of the
extent of the type vocabulary.
For this reason, we lift the notion of CSUs in
Definition~\ref{def:csu} to encompass unification of terms containing
type variables, where type variables are treated as ``frozen''
type-level constants.
Since there is an injection from the concrete types that can be
constructed under any possible extension to the existing type
vocabulary to the types that can be constructed from the existing type
vocabulary plus the infinite set of type variables, it suffices to
consider CSUs that are ``type-generic'' in the latter context.
\begin{definition}\label{def:type-generic-csu}
A \emph{type-generic complete set of unifiers} for two schematic terms $t_1$
and $t_2$ of identical type that are parameterized by the type
variables $A_1,\ldots,A_n$ 
is a collection $\Theta$ of substitutions such that
$\{\apptysubst{\tau_1/A_1,\ldots,\tau_n/A_n}{\theta} 
\sep \theta \in \Theta\}$ is a complete set of unifiers for
$\apptysubst{\tau_1/A_1,\ldots,\tau_n/A_n}{t_1}$ and
$\apptysubst{\tau_1/A_1,\ldots,\tau_n/A_n}{t_2}$ for any type
expressions $\tau_1,\ldots,\tau_n$.
We write \gencsu{t_1}{t_2} to (ambiguously) denote such a collection
$\Theta$.  
\end{definition}

We now describe conditions under which we can identify a schematic
version of case analysis over an atomic formula.
%
\begin{definition}\label{def:amenable}
Let ${\cal S} = (\schmsequent{\Psi}{\Si}{\G, A}{D})$ be a schematic
sequent in which $A$ is an atomic formula and let
\[{\cal C} =
(\schmclause{\Psi'}{\seq{x:\alpha}}{A'}{B})\]
be a schematic clause
named away from $\Si$ and such that  $\Psi'$ is disjoint from $\Psi$.
Then $\cal S$ is \emph{analyzable in a generic way with respect to the
clause $\cal C$ on $A$} and the analysis produces the corresponding set
$\gendefcases{\cal S}{\cal C}$ if one of the following conditions hold:
\begin{enumerate}
\item
  $A$ and $A'$ are not unifiable under any instantiation of type
  variables. 
  In this case, $\gendefcases{\cal S}{\cal C} = \emptyset$.

\item There is a type substitution $\phi$ for the variables in $\Psi'$
  whose range is well-formed in $\Psi$ and there is a type generic CSU
  $\Theta$ for $A$ and $\apptysubst{\phi}{A'}$ such that, for any type
  substitution $\phi'$ that makes $\apptysubst{\phi'}{A}$ and $\apptysubst{\phi'}{A'}$
  unifiable, there is a type substitution $\delta$ such that
  $\phi' = \tysubstcomp{\delta}{\phi}$.
  In this case, $\gendefcases{\cal S}{\cal C}$ denotes
  the set $
  \{ \schmsequent{\Psi}{\Sigma \theta}{\Gamma\theta, \appsubst{\theta}{\apptysubst{\phi}{B}}}{D\theta}\ \sep\ \theta \in \Theta \}.
  $

\end{enumerate}
$\cal S$ is \emph{amenable
to case analysis on $A$ against a schematic definition $\cal D$} if
it is analyzable with respect to every clause in $\reduceddef{\cal D}$
on $A$ and the case analysis then results in 
the set of sequents
\[\gendefcases{\cal S}{\reduceddef{\cal D}} = \{
{\cal S'} \sep\ {\cal C} \in \reduceddef{\cal D}, {\cal S'} \in
\gendefcases{\cal S}{\cal C}  \}.\]
\end{definition}
The conditions governing the analyzability of a schematic sequent with
respect to a schematic clause can be understood as follows.
If the first condition is satisfied, then we can use a generic rule
with an empty set of premises deriving from the clause being
considered.
If the second condition is satisfied and we can find a type
substitution and a type generic CSU satisfying its requirements, we
can once again deal generically with this clause: the premise sequents
that need to be considered in a \defL rule for a type instance of the conclusion
sequent in \Gee will be a type instance of the set of sequents
produced by the analysis.  

Using the above definition, a schematic version of the definition
left rule can be formulated as follows:  

\smallskip

\[
\begin{array}{c}
\infer[\schmdefL]
      {\schmsequent{\Psi}{\Sigma}{\Gamma, A}{F}}
      {\gendefcases{\schmsequent{\Psi}{\Sigma}{\Gamma, A}{F}}{\reduceddef{\cal D'}}}\\
{\cal D'}\ \mbox{\rm is a variant of}\ {\cal D}\ 
  \mbox{\rm named away from $\Sigma$ and $\Psi$}
\end{array}
\]

\smallskip

\noindent The rule is governed by a proviso: it can be used only if the
lower sequent is amenable to case analysis on $A$ against
the operative (schematic) definition $\cal D$. 

In Section~\ref{sec:schm-abella}, we will illustrate how this rule can
be used in constructing schematic proofs.
In understanding the content of the rule, it is useful
to also see situations in which it is \emph{not} applicable. 
Suppose that we have the constant $\kp:\tyschm{A}{A \to \iota}$ 
and the clause
$\clause{x:A}{\cinst{\keq}{A}\app x\app x}{\top}$ for the predicate
  $\keq:\tyschm{A}{A \ra A \ra \prop}$ and we want to prove the formula
\[
\schmthm{A,B}{\rforallx{(x:A) (y:B)} {\cinst{\keq}{\iota}\app
    (\cinst{\kp}{A}\app x)\app (\cinst{\kp}{B}\app y) \rimp \rfalse}}
\]
This task can be reduced to proving the following schematic sequent:
\[
\schmsequent{A,B}{x:A,y:B}
            {\cinst{\keq}{\iota}\app (\cinst{\kp}{A}\app x)\app (\cinst{\kp}{B}\app y)}
            {\rfalse}
\]
We might want to apply the \schmdefL to the only assumption atom here
but unfortunately neither of the conditions in
Definition~\ref{def:amenable} holds.
On the one hand, the atomic formula is unifiable with the head of the
clause for \keq in the case that $A$ and $B$ are set to the same
type.
On the other hand, a type generic CSU does not exist for the atomic
formula and a type instance of the head of the clause for \keq because
the unifiability of their instances depends on  whether or not the
same types are substituted for $A$ and $B$.

\subsection{The schematic induction rule}
\label{subsec:schm-ind-rules}
A schematic definition block is said to be \emph{accommodative to
  induction} if every schematic clause within the block is
parameterized by an empty sequence of type variables.
If a block has this character, then each of its schematic clauses
gives rise to exactly one clause in the definition block in \Gee that
is generated by instantiating the type variables parameterizing the
block with concrete types.
Thus, every instance of such a schematic definition block has a finite
number of clauses and hence all these clauses can be designated as
inductive ones.
In consonance with this observation, we allow a schematic definition
block that is accommodative to induction to be designated as
inductive.

A schematic version of the induction rule $\indL$ is then given as
follows:

\smallskip

\begin{center}
\small
\begin{tabular}{c}
  \infer[\schmindL]
        {\schmsequent{\Psi}{\Sigma}{\Gamma, \cinst{p}{\seq{\tau}}\app \seq{t}}{F}}
        {
          \begin{array}{c}
          \{\schmsequent{\Psi}{\seq{x :\alpha}}{B[S/p]}{S\app \seq{t'}}\ \sep\ 
          \clause{\seq{x : \alpha}}{\cinst{p}{\seq{\tau}}\app \seq{t'}}{B} \in \binst{\cal B}{\seq{\tau}}\}
          \\
          \schmsequent{\Psi}{\Sigma}{\Gamma, S \app \seq{t}}{F}
          \end{array}
        }
    \\[1pt]
    \begin{tabular}{l}
      The schematic block $\cal B$ for $p$ is inductive\\
      and has only $p$ associated with it
    \end{tabular}
\end{tabular}
\end{center}

\smallskip

\noindent 
This rule essentially parameterizes the conclusion and
premises of the \indL rule with $\Psi$.
As in the case of the \indL rule, the type of the term that
instantiates $S$ must be identical to that of
$\cinst{p}{\seq{\tau}}$.
Moreover, the type variables appearing in this term must be contained
in $\Psi$.
It is easy to see that this rule reduces to \indL under the
instantiation of $\Psi$ with ground types. 
Although the rule that we have presented here is applicable only to
schematic definition blocks that have defining clauses for exactly one
predicate, it can be generalized to deal with predicates that are defined
mutually inductively.
We elide a discussion of the generalized rule in this paper.

\subsection{Soundness of the schematic proof system}
\label{subsec:soundness}

The schematic nature of the lifted proof system we have presented can
be articulated via a soundness theorem.

\begin{theorem}\label{thm:schm_sound}
  Let $\Pi$ be a derivation of the schematic sequent
  \[\schmsequent{\Psi}{\Si}{\G}{B}\]
  that is constructed using the 
  schematic proof rules we have described. If $\Phi$ is a substitution
  that maps each variable in $\Psi$ to a ground type, then $\Pi[\Phi]$
  is proof in \Gee for the sequent
  $\sequent{\Si[\Phi]}{\G[\Phi]}{B[\Phi]}$. 
\end{theorem}
\begin{proof} Only a sketch is provided here.
  The proof proceeds by induction on the height of $\Pi$, 
  considering the different possibilities for the last (schematic)
  inference rule in the derivation.
  The argument follows an obvious pattern in most cases: We 
  invoke the induction hypothesis on the premises of the last rule to
  determine that the type instances of the premise derivations will be
  derivations in \Gee of the type instances of the corresponding
  premises.
  These derivations can then be extended by using a type instance
  of the last rule in the schematic proof to obtain a derivation in
  \Gee of the desired type instance of the concluding sequent in
  $\Pi$; when the last rule is either $\schmsomeR$ or $\schmallL$, we
  will need the additional observation here that well-typedness of terms is
  preserved under type instantiation.
  It is now easy seen that the proof that has been generated in
  this way is itself the relevant type instance of $\Pi$.

  The only case that needs further elaboration is that 
  when the last rule is $\schmdefL$.
  Given any schematic definitional clause, the proviso of this rule
  ensures that if the matching between the atomic formula being
  analyzed with the head of the 
  clause fails (\ie, condition (1) of Definition~\ref{def:amenable} is
  satisfied), then it will also fail under the instantiation of type
  variables. The proviso also ensures that when the matching succeeds
  (\ie, condition (2) of Definition~\ref{def:amenable} is satisfied),
  it also succeeds under the type instantiation. Moreover, the type
  generic natural of the matching ensures the structure of the premise
  generated from the matching is preserved under the type
  instantiation. These observations allow us to complete the argument
  in this case as well using the previously described pattern.
\end{proof}
The proof of the above theorem is constructive. Its procedural
interpretation provides us a function for constructing proofs in \Gee
from the schematic proofs.
A consequence of the theorem is that if a schematic formula
$[A_1,\ldots,A_n]F$ has a schematic proof, then, for any 
ground types $\tau_1,\ldots,\tau_n$,
$F[\tau_1/A_1,\ldots,\tau_n/A_n]$ is provable is \Gee.



\section{Constructing Schematic Proofs}\label{sec:schm-abella}

We have implemented an extension to \Abella that supports schematic
polymorphism based on the ideas described in this paper. 
We present some examples below that illustrate the new capabilities of
this system and also highlight some of the limitations of the form of
polymorphism it realizes.

We first consider a schematized version of the \gappend
predicate from Section~\ref{subsec:proofs-in-abella}.
Associating the type schema
$\tyschm{A}{\klist\app A}$ with \nil and $\tyschm{A}{A \ra \klist\app A \ra
  \klist\app A}$ with $\scons$,
this predicate is defined by the following clauses in a
block with the predicate signature $\kapp:\tyschm{A}{\klist\app A \to
  \klist\app A \to \klist\app A \to \prop}$:
\begin{tabbing}
\quad\=\quad\=\kill
\>$\clause{l:\klist\app A}{\cinst{\gappend}{A}\app \nil\app l\app l}{\top}$\\
\>$\clauseheader{x:A}\clauseheader{l_1,l_2,l_3:\klist\app A}$\\
\>\>$\clausesans{\cinst{\gappend}{A}\app (x\scons l_1)\app l_2\app (x\scons l_3)}
    {\cinst{\gappend}{A}\app l_1\app l_2\app l_3}$
\end{tabbing}
Since neither of the clauses is parameterized by a type, this block
can be designated as inductive.
We can get definition blocks in \Gee
by instantiating the variable $A$ in this
schematic definition with a concrete type. 
Note that the blocks so generated are independent of each other and
also finite in size.
As such, the resulting ``instance'' definitions for \gappend can also
be treated as inductive in \Gee.

In the context of the schematic definition, we can write the following
schematic formula that encodes the discussed functional property of
\gappend:
\begin{tabbing}
\quad\=\qquad\=\kill
\>$[A]\rforallxheader{l_1,l_2,l_3,l_4:\klist\app A}$\\
\>\>${\cinst{\gappend}{A}\app
  l_1\app l_2\app l_3 \supset \cinst{\gappend}{A}\app l_1\app l_2\app l_4 \supset
            l_3 = l_4}$.
\end{tabbing}
The \schmindL rule provides the basis for lifting the annotated style
of induction discussed in Section~\ref{subsec:proofs-in-abella} to the
schematic setting. 
Letting $F$ denote the formula 
\begin{tabbing}
\quad\=\qquad\=\kill
\>$\rforallx{l_1,l_2,l_3,l_4:\klist\app A}{\cinst{\gappend}{A}^*\app
  l_1\app l_2\app l_3 \supset \cinst{\gappend}{A}\app l_1\app l_2\app l_4 \supset
            l_3 = l_4}$,
\end{tabbing}
using the induction tactic will allow us to reduce the task of proving
the formula of interest to constructing a derivation for the schematic
sequent
\begin{tabbing}
\quad\=\qquad\=\kill
\>$\schmsequentheader{A}{(l_1,l_2,l_3,l_4:\klist\app A)}$\\
\>\>$\schmsequentbody{F, \cinst{\gappend}{A}^@\app l_1\app l_2\app l_3, 
                  \cinst{\gappend}{A}\app l_1\app l_2\app l_4}{l_3 = l_4}.$
\end{tabbing}
This sequent is amenable to case analysis on
$\cinst{\gappend}{A}^@\app l_1\app l_2\app l_3$ against the operative
definition: it is easy to see that there is a type generic CSU for the
atom with the head of the reduced form of each \gappend clause
instantiated with $A$, 
thereby satisfying the second condition in
Definition~\ref{def:amenable}.
The remainder of the proof follows the structure of the one sketched
in Section~\ref{subsec:proofs-in-abella}.
Note that instantiating the type variable $A$ with $\iota$ in this
schematic proof will in fact yield the earlier proof. 

The next example concerns the append relation defined in the
specification logic.
Definite clauses can also be schematized as described, e.g., in
\cite{nadathur92typesshort}.
Encoding such definite clauses for \append yields the following
schematic clauses for \prog: 
\begin{tabbing}
\quad\=\qquad\=$\prog\app$\=$(\cinst{\append}{A}\app (x\scons l_1)\app l_2\app (x \scons l_3))$\quad\=\kill
\>$\schmclause{A}{l:\klist\app A}{\prog\app (\cinst{\append}{A}\app \nil\app l\app l)\app \strue}{\top}$\\
\>$[A]\clauseheader{x:A,l_1:\klist\app A, l_2:\klist\app A, l_3:\klist\app A}$\\
\>\>$\prog\app (\cinst{\append}{A}\app (x\scons l_1)\app l_2\app (x \scons l_3))$\\
\>\>\>$(\atmconst\app (\cinst{\append}{A}\app l_1\app l_2\app l_3))$\>$\rdef \quad \top$
\end{tabbing}
%
%
The type schema $\tyschm{A}{\klist\app A \ra \klist\app A \ra
  \klist\app A \ra \omic}$ is associated with \append here. 
Each type instance of a clause above yields a \prog clause in \Gee
that encodes the corresponding type instance of an \append clause in
the specification logic.
Combining the schematic clauses with the ones for \prove seen earlier
gives us a (schematic) encoding of the specification logic
definition.
The functional nature of the ``polymorphic'' definition of \append in
the specification logic can now be expressed by the following
schematic formula:
\begin{tabbing}
  \quad\=\qquad\=\kill
  \>$[A]\rforallxheader{l_1,l_2,l_3,l_4:\klist\app A}$\\
  \>\>${\sprove{\cinst{\append}{A}\app l_1\app l_2\app l_3} \supset
    \sprove{\cinst{\append}{A}\app l_1\app l_2\app l_4} \supset l_3 = l_4}.$
\end{tabbing}
A schematic proof can be constructed for this formula, from which the
one discussed in Section~\ref{subsec:proofs-in-abella} can be obtained
as a type instance; we omit the details that should be easy to fill in
given the discussions in Section~\ref{subsec:encode-spec-logic}.

Although many useful schematic theorems can be established using our
schematic proof system, there are some that lie beyond its
capabilities. 
Towards understanding the content of this observation, let us consider
the schematic formula 
%
\begin{tabbing}
  \quad\=\qquad\=\kill
\>$\schmthm{A,B}{\rforallx{x:A, f:A \to B}{\rexistsx{y:B}{}}}$\\
\>\>$(\cinst{\keq}{\iota}\app (\cinst{\kp}{A}\app x)\app (\cinst{\kp}{B}\app y))
             \ror ((\cinst{\keq}{\iota}\app (\cinst{\kp}{A}\app x) \app
             (\cinst{\kp}{B}\app y)) \rimp \rfalse)$
\end{tabbing}
which uses the predicates \keq and \kp introduced in
Section~\ref{subsec:schm-def-rules}. 
We see that every type instance of this formula must hold based on the
following reasoning.
When $A$ and $B$ are instantiated by the same type, we instantiate the
existentially quantified variable $y$ with the same value that
instantiates the universally quantified variable $x$ and then observe
that the left branch of the disjunction holds.
If $A$ and $B$ are instantiated with different types then
corresponding to each term $t$ that is chosen for $x$ we pick the term
$(f\app t)$ for $y$ and then observe that the right branch of the 
disjunction must hold.
%
However, the structure of above argument depends on the type instance
under consideration and hence it cannot be captured by a schematic proof.
Moreover it can be checked that it is impossible to provide a
schematic proof for this formula.

\section{Covering the Full Logic}
\label{sec:fulllogic}

In the discussions up to this point, we have considered only a subset
of the specification logic \HH and, correspondingly, of the logic
\Gee.
We did this so that we could focus on the main technical issues
relating to our ideas for introducing schematic polymorphism.
The features of \HH and \Gee that we have left out of the discussion
are, however, central to the usefulness of \Abella in its
application domain: in particular, they are vital to the
logical treatment of higher-order abstract syntax.
It is important therefore that the techniques described for the
simplified logics be extendable to the fully featured logics.
This is indeed the case, as we try to demonstrate in this section.
A complete development for the full logic that elaborates on the ideas
discussed here can be found in Wang's doctoral
thesis~\cite{wang16phd}.
Our implementation of ``schematic \Abella'' is, in fact, based on this
complete development.

\subsection{Additional logical features of \HH and \Gee}

The features that we have omitted in the earlier discussions are the
$\nabla$ quantifier in \Gee and the hypothetical and generic forms for
goals in \HH.
These features conspire to provide an inductive treatment of
syntax even in the presence of binding constructs.
We describe the richer forms to the logics and also motivate their
usefulness below.

The enrichment to Horn clauses that results in \HH is easy to describe
at the syntactic level: goal formulas are allowed to contain
universal quantifiers and implications of the form $D \simply G$ where
$D$ is a definite clause.
Note that this change impacts also the syntax of definite clauses
which, in the enriched form, are referred to as hereditary Harrop formulas.
Permitting universal quantifiers in goals implies that the vocabulary for
constructing terms might change in the course of a derivation.
This aspect is accounted for by changing the form of the specification
logic sequent to $\hhsequent{\Xi}{\Gamma}{G}$, where $\Xi$ represents
an eigenvariable context.
The earlier present rules for the specification logic are modified in
an obvious way to take into account the richer structure for
sequents.
We additionally have the following rules to treat the new forms for
goal formulas:

\smallskip

\begin{center}
\begin{tabular}{cc}
\infer[\impR]
      {\hhsequent{\Xi}{\Gamma}{D \simply G}}
      {\hhsequent{\Xi}{\Gamma, D}{G}}
&
\infer[\forallR]
      {\hhsequent{\Xi}{\Gamma}{\forallx{x:\alpha}{G}}}
      {\hhsequent{\Xi,x:\alpha}{\Gamma}{G}}
\end{tabular}
\end{center}

\smallskip

The new forms for goals in \HH add a dynamic character to both the
signature and the program context that is relevant in a derivation.
This feature turns out to be quite useful in capturing recursion over
binding constructs.
When carrying out an analysis over the body of such a construct, it is
necessary to distinguish occurrences of the bound variable and it may
also be necessary to attribute specific properties to the bound
variable for the duration of the analysis.
These are exactly the capabilities provided by universal quantifiers
and implications in goals.
We refer the reader to \cite{miller12proghol}, amongst other sources,
for concrete examples that use these observations in constructing
specifications related to syntactic structure in which binding is an
important component.

In reasoning about syntax, it is often necessary to treat entities
such as bound variables as atomic, unanalyzable components.
The universal quantifier provides this ability in \HH.
However, this quantifier cannot be used for a similar purpose in \Gee
because the logic gives it an extensional reading; this interpretation
is manifest, for example, in the \defL rule which considers the different 
ways in which eigenvariables may be instantiated towards making an
assumption formula hold.
The $\nabla$ quantifier in \Gee can be understood as a new ``generic''
quantifier that fills this gap.\footnote{This quantifier has other
  uses, such as in encoding the uniqueness of names. We defer a more
  detailed discussion of these aspects to, \eg, \cite{baelde14jfr}.}
Formally, we may construct a proof of a sequent in which the formula
$\nablax x {(B\app x)}$ occurs on the right by picking a new constant
$c$ that does not appear in the sequent and then proving the sequent
that results from replacing $\nablax x {(B\app x)}$ with $(B\app
c)$.
The constants that are to be used in this way in proofs belong to a
special category called the \emph{nominal constants}.
The treatment of $\nabla x. (B\app x)$ when it appears on the left
side of the sequent, \ie, as an assumption formula, is symmetric: we
get to use $(B\app c)$ where $c$ is a fresh nominal constant as an
assumption instead.

In addition to their use in formulas, the $\nabla$ quantifier can also
be used in the head of a clause in a definition in \Gee~\cite{gacek11ic}.
More specifically, the full form for such clauses is the following:
\[ \clause{\seq{x:\alpha}}{(\nablax {\seq{z:\tau}}A)} {B}. \]
In generating instances of such a clause, the $\nabla$ quantifiers
over the head may be instantiated by any collection of distinct
nominal constants of suitable types.
Further, the universal quantifiers over the clause may be instantiated
by arbitrary terms of the requisite types with the proviso that they
must not contain any of the nominal constants used to instantiate the
$\nabla$ quantifiers.
Thus, the order of the quantifiers facilitates the encoding of
dependency information.
For example, consider the (extended) clause
\[\clause{t : \iota}
         {(\nablax {x:\iota} {\kfresh\app x\app t})} {\rtrue};\]
$\kfresh$ is assumed to be a predicate of type $\iota \to \iota \to
\prop$ here.  
This clause codifies the requirement that $(\kfresh\app x\app t)$
holds exactly when $x$ is a nominal constant that does not appear in
the term $t$.
This capability has a general use in the context of reasoning about
syntactic structures in the presence of binding: it allows us to make
explicit properties such as the distinctness of variable occurrences
that are captured by different binders and the uniqueness of
assignments to such variables when these variables are represented by
nominal constants.

\subsection{Schematization applied to the full system}
\label{subsec:schm-fulllogic}

The additional features in \HH and \Gee interact in a benign way with
our ideas related to parameterization based on types.
In fact, parameterization provides a more systematic treatment of one
aspect in comparison with the simply typed version of the system.
This aspect concerns the encoding of \HH in \Gee.
We discuss these matters below.

Following the lines described in
Section~\ref{subsec:encode-spec-logic}, the encoding of \HH is
realized by lifting its signature into \Gee and by capturing its
derivability relation in a definition.
In encoding the derivability relation, we now also have to
account for the changing nature of the signature and the program
context. 
Modelling the changes to the signature is simplified by the presence
of the $\nabla$ quantifier in \Gee: this quantifier can be used to
introduce a nominal constant that implicitly encodes the desired
signature enhancement.
To deal with the former aspect, we augment the \prove predicate with
an additional parameter that represents a list of the specification
logic clauses that are added dynamically in the process of searching
for a derivation.

Playing the above ideas out in detail leads to \prove now having the
type $\typenat \ra (\klist\tyapp \omic) \ra \omic \ra
\prop$.\footnote{We assume that the type constructor for lists and the
  (schematic) list constructors have been added to the vocabulary in
  this discussion.}
%
Further, the clauses defining this predicate change to the following:
\begin{tabbing}
\quad\=\quad\=\kill
\>$\clause{d:\omic, l: \klist\tyapp \omic}{\member\app d\app (d ::_\omic
  l)}{\top}$\\ 
\>$\clause{d, d':\omic, l: \klist\tyapp \omic}{\member\app d\app (d'
  ::_\omic l)}{\member\app d\app l}$\\[5pt]
\>$[A]\clauseheader{n:\typenat,l:\klist\tyapp \omic, d:A \ra \omic,a:\omic,t:A}$\\
\>\>$\clausesans{\backchain\app n\app l\app (\sfall_A\app d)\app a}
                {\backchain\app n \app l\app (d\app t)\app a}$\\
\>$\clauseheader{n:\typenat,l:\klist\tyapp \omic, g:\omic,a:\omic}$\\
\>\>$\clausesans{\backchain\app n\app l\app (g\simply a)\app a}
                {\prove\app n \app l\app g}$\\[5pt]
\>$\clause{n:\typenat, l : \klist\tyapp \omic}{\prove\app n\app l\app \strue}{\top}$\\
\>$\clauseheader{n:\typenat,l:\klist\tyapp \omic,g_1:\omic,g_2:\omic}$\\
\>\>$\clausesans{\prove\app (s\app n)\app l \app (g_1 \sconj g_2)} 
                {\prove\app n\app l \app g_1 \rand \prove\app n\app l \app g_2}$\\
\>$\clauseheader{n:\typenat,l:\klist\tyapp \omic,d:\omic,g:\omic}$\\
\>\>$\clausesans{\prove\app (s\app n)\app l \app (d \simply g)}                
                {\prove\app n\app (d ::_\omic l)\app g}$\\
\>$[A]\clauseheader{n:\typenat,l:\klist\tyapp \omic,d:A \ra \omic}$\\
\>\>$\clausesans{\prove\app (s\app n)\app l\app (\sfall_A\app d)}                
                {\nablax{x:A}{\prove\app n\app l\app (d\app x)}}$\\
\>$\clauseheader{n:\typenat,l:\klist\tyapp \omic,a:\omic, g:\omic}$\\
\>\>$\clausesans{\prove\app (s\app n)\app l \app (\atmconst\app a)} 
     {(\prog\app a \app g) \rand (\prove \app n\app l \app g)}$\\
\>$\clauseheader{n:\typenat,l:\klist\tyapp \omic,a:\omic}$\\
\>\>$\clausesans{\prove\app (s\app n)\app l\app (\atmconst\app a)}                
                {\member\app d\app l \rand \backchain\app n\app l\app
                  d \app a}$
\end{tabbing}
The constants $\backchain$ and $\member$ used here have the types
$\typenat \ra (\klist\tyapp \omic)\ra \omic\ra \omic\ra \prop$ and
$\omic\ra (\klist \tyapp \omic) \ra \prop$, respectively, and $\sfall$
has the schematic type $[A](A \ra \omic) \ra \omic$. 
The clauses for $\backchain$ encode backchaining on the definite clauses
added dynamically during the derivation.

The clause for \prove that treats universal goals and the clause for \backchain that
corresponds to instantiating a specification logic clause illuminate
the usefulness of the schematic polymorphism developed in this paper
in realizing a hygienic encoding of \HH. These clauses must treat
quantification over variables of all possible types. The ability to
parameterize the clause by the type of the variable results in a
precise encoding of this fact, an improvement over the existing
version of \Abella in which this aspect is treated in an ad hoc way in
the implementation.

The addition of $\nabla$ to the collection of logical symbols of
$\Gee$ does not complicate the description of schematic proof rules
for these symbols.
The only remaining aspect is the schematization of the rules
pertaining to the enriched form for definitions.
Here we build, again, on the version of the rules described for the
simply typed case.
The schematic version of the (appropriate) \defR and induction rules
are easily obtained.
To consider the \defR rule as an example, in matching the atomic
formula in the sequent with the head of a clause, we now have to also
consider substitutions for the $\nabla$ quantified variables in the
head.
However, this is just as in the simply typed case and it is dealt with
in a way that is orthogonal to the treatment of type variables.

There are more details to address in schematizing the \defL rule for
the extended form for definitions but this can be done as we now
indicate.
One requirement is that we must generalize the definition of
type-generic CSUs to terms containing nominal constants.
The critical observation in doing this is that we can treat
nominal constants whose types might contain type variables in much the
same way that we have treated ordinary constants.
Once we have determined this, we are able to use the original
definition for type-generic CSUs
(Definition~\ref{def:type-generic-csu}) in the new context. 
Next, we need to generalize Definition~\ref{def:amenable}, which
describes the analysis of a sequent ${\cal S}$ with respect to a
clause ${\cal C}$, to the case where the clause has the (more general)
form $(\schmclause{\Psi'}{\seq{x:\alpha}}{(\nablax{\vec{z}}{A'})}{B})$.
This is realized through two steps.
First, we describe the enumeration of cases based on the possible
forms of $A'$ resulting from instantiating $\vec{z}$ with nominal
constants, with the possibility that these nominal constants may be used
to instantiate variables in $\cal S$.
This kind of enumeration and its algorithmic treatment through
unification has been considered elsewhere~\cite{gacek11ic} and it is,
in any case, orthogonal to the schematization of \defL.
Second, for each case generated in this way, we require that $S$
be analyzable in a generic way with respect to $\cal C$ in the
sense that condition (1) or (2) in Definition~\ref{def:amenable} is
satisfied.
Once Definition~\ref{def:amenable} has been generalized through these
two steps, it can be used just as in Section~\ref{sec:schm-proofs} 
to describe a schematized version of the \defL rule for definitions in
the full form that is permitted for them in \Gee.

\section{Related Work and Conclusion}
\label{sec:comparison}

The approach we have used in this paper to realize polymorphism by
parameterizing the types of variables, constants and predicates has
been influenced by previous work on functional languages (\eg, see
\cite{damas82,milner78jcss}) and logic programming languages (\eg, see
\cite{hanus89iclpshort,laksham91ilps,mycroft84ai,nadathur92typesshort}).
At a spiritual level, our development is closest to that in
\cite{nadathur92typesshort}, where the parameterization is treated as
a schematic one, to be understood eventually by a translation into a
simply typed (logical) language. 
Despite the similarities in the underlying idea, we note that the context
of its application and, hence, the technical challenges in making it
work are significantly different.
The focus in a logic programming language, for example, is on using a
logic for computation, whereas the main emphasis in this paper has been
to get schematic polymorphism to work in constructing proofs in a
logic that includes case analysis over fixed-point definitions.
Moreover, we have had the additional concern of enabling reasoning in
this logic over programs in a logic programming language that itself
supports schematic polymorphism.


There are other systems that support higher-order abstract syntax and
that provide reasoning capabilities similar to \Abella.
Two specific systems that do this are  Twelf~\cite{pfenning99cade} and
Beluga~\cite{pientka10ijcarshort}.
These systems currently lack support for polymorphism.
We believe that the ideas presented in this paper can be used to endow
them as well with this feature.
To understand this point, note that the only major difficulty in
implementing schematic polymorphism in Twelf or Beluga lies in the
formalization of case analysis in a way that is not sensitive to type
information.
Note also that, in order to be able to provide faithful encodings of
object systems, any such analysis must be based on the structures of
terms.
With this understanding of the role of terms and some reflection on
the intended fixed-point reading of specifications in these systems,
it becomes clear that unification plays a vital role in realizing case
analysis.
However, unification in the context of typed $\lambda$-calculi is
typically sensitive to typing information and this turns out to be the
main challenge to carrying out case analysis in a schematic way when
type parameterization is permitted.
The key technical novelty of our work is to treat type
variables as unknown entities (or ``black boxes'') and to use the
result of a unification computation only when it can be produced
without obtaining any further information about these type
entities.
Using this approach, it should be possible to extend case analysis to
versions of these other systems that permit type parameterization.
We believe that this is a direction that is worthy of further
exploration because of the light-weight way in which it is able to
support a useful form of polymorphism; the method piggy-backs on
existing mechanisms and is therefore implementable with limited
additional effort. 

An alternative approach to achieving polymorphism is through a module
system and instantiation of modules. Rabe and
Sch{\"u}rmann~\cite{rabe2009} developed a module system for LF which
is the specification language of Twelf. The idea is to group LF
declarations and definitions into \emph{signatures}. To make use of a
signature in another context, they apply to the signature
\emph{signature morphisms} which are mappings from constants in the
signature to terms in the target context. Like our work, their module
system preserves the logical foundation of Twelf because all the extra
devices introduced to implement the module system are elaborated into
the core LF. For reasoning they simply use the mechanisms provided by
\Twelf such as coverage checking on these elaborated
definitions. However, as they have noted in their paper, signature
morphisms may not preserve the validity of the checking
process. Therefore, a theorem proved in a given module may not hold under
an instantiation of that module. We observe that coverage checking in
Twelf, similar to case analysis in Abella, makes use of
unification. Therefore, we believe that the idea of schematic
polymorphism that we have developed can be used to generalize coverage
checking in a way that is stable under signature morphisms. 

In general-purpose theorem provers such as Coq and Isabelle, 
reasoning principles such as case analysis and induction are derived
from mechanisms for defining inductive data types (in Coq) or
algebraic data types interpreted as inductively defined sets (in
Isabelle/HOL).
By parameterizing such definitions with type variables, we can obtain
a ``type-generic'' way to support case analysis and induction.
Because of these differences, especially because the reasoning
principles are not dependent in a fundamental way on the inductive
structure of terms, our work on realizing schematic polymorphism does
not translate directly to the context of these general-purpose theorem
provers.
However, there has been work on exploiting the benefits of
higher-order abstract syntax and the two-level logic approach within
general-purpose theorem provers by encoding a specification logic that
supports higher-order abstract syntax as a library and then reasoning
about specifications written in this logic through the encoding.
One example of such work is embodied in the Hybrid 
system~\cite{felty12jar}.
Supporting polymorphism in this kind of a setting will eventually
become important and we believe that the ideas we have developed in
this paper will be relevant to these systems too at that point.


As already noted, we have implemented a version of \Abella
that incorporates schematic polymorphism based on the ideas developed
in this paper. 
This version provides several benefits over the previously available
version of the system.
First, it allows us to write polymorphic specifications in \HH that
can be executed as programs in \LProlog using previously existing
capabilities and that can now be reasoned about using \Abella.
This support for polymorphism eliminates a lot of redundancy in
programming and also makes the programs or specifications and proofs
of their properties more modular.
Second, even without considering the specification logic, \Abella
itself is able to treat polymorphic definitions and theorems.
Third, the extension means we are able to create standard libraries
for \Abella that contain definitions of generic data structures such
as lists and sets and theorems describing their properties.
We have only recently completed the implementation of our system and
it has therefore not seen sufficient use for us to quantify the above
benefits based on actual experience.
However, we have used the system in one significant project that
concerns the implementation and verification of a compiler for an
extension of PCF that includes
recursion~\cite{wang16phd,wang16esopshort}.  
The polymorphism enhancements helped eliminate redundancies in both
the implementation and the verification in the manner described.
For example, using it led to a reduction of $14\%$, measured in terms
of the lines of code, in the formalization of the library components of the
project in comparison with an alternative development in the earlier
version of \Abella that did not support polymorphism.
More details about this project and comparisons can be found at the
website mentioned in the introduction.
Information on how to download the enhanced system can also be found
at this website. 

\section*{Acknowledgements}

This work has been supported by the National Science Foundation grant
CCF-1617771. Opinions, findings and conclusions or
recommendations that are manifest in this material are those of the
participants and do not necessarily reflect the views of the NSF.

\bibliographystyle{ACM-Reference-Format}
\bibliography{../../../references/master,../ppdp18/local}

\end{document}